\documentclass[12pt, draftclsnofoot, journal, onecolumn]{IEEEtran}
\usepackage[T1]{fontenc}

\usepackage[dvips]{graphicx}
\usepackage{amsmath, amsthm, amssymb, amsfonts, mathtools, array, bm}
\usepackage{cases, multirow, cool, srcltx,color, cite, url}
\definecolor{marron}{RGB}{60,30,10}
\definecolor{darkblue}{RGB}{0,0,80}
\definecolor{lightblue}{RGB}{80,80,80}
\definecolor{darkgreen}{RGB}{0,80,0}
\definecolor{darkgray}{RGB}{0,80,0}
\definecolor{darkred}{RGB}{80,0,0}
\definecolor{shadecolor}{rgb}{0.97,0.97,0.97}

\usepackage{algorithm, algpseudocode, cool, xcolor, xargs}    
\usepackage{lipsum, courier}
\usepackage{silence} 
\usepackage{enumerate, afterpage,dblfloatfix}
\usepackage{lettrine}
\input Acorn.fd

\usepackage{fourier-orns}



\usepackage[T1]{fontenc} 
\usepackage[cmintegrals]{newtxmath} 

\ifCLASSOPTIONcompsoc
 \usepackage[caption=false,font=normalsize,labelfont=sf,textfont=sf]{subfig}
\else
 \usepackage[caption=false,font=footnotesize]{subfig}
\fi

\makeatletter
\newcommand*\Bigcdot{\mathpalette\Bigcdot@{.5}}
\newcommand*\Bigcdot@[2]{\mathbin{\vcenter{\hbox{\scalebox{#2}{$\m@th#1\bullet$}}}}}
\makeatother

\newtheorem{thm}{Theorem}

\newtheorem{corollary}{Corollary}

\newtheorem{remark}{Remark}

\def\figref#1{Fig.~\ref{#1}}

\newmuskip\pFqmuskip

\newcommand*\pFq[6][8]{%
  \begingroup 
  \pFqmuskip=#1mu\relax
  \begingroup\lccode`\~=`\,
  \lowercase{\endgroup\let~}\pFqcomma
  {}_{#2}F_{#3}{\left(\left. \genfrac..{0pt}{}{#4}{#5} \right\vert #6\right)}%
  \endgroup
}
\newcommand{\pFqcomma}{\mskip\pFqmuskip}

\newmuskip\pFqskip
\mathchardef\pFcomma=\mathcode`, 

\begin{document}

\title{A Generalized Fading Model with Multiple Specular Components}

\author{Young Jin Chun
\thanks{Y. J. Chun is with the Wireless Communications Laboratory, ECIT Institute, Queen's University Belfast, Belfast, BT3 9DT, U.K. (e-mail: ychun@ieee.org).}%
}

\maketitle

\vspace{-10 mm}

\begin{abstract}
The wireless channel of 5G communications will have unique characteristics that can not be fully apprehended by the traditional fading models. For instance, the wireless channel may often be dominated by finite number of specular components, the conventional Gaussian assumption may not be applied to the diffuse scattered waves and the point scatterers may be inhomogeneously distributed. These physical attributes were incorporated into the state-of-the-art fading models, such as the $\kappa$-$\mu$ shadowed fading model, the generalized two-ray fading model and the fluctuating two ray fading model. Unfortunately, much of the existing published work commonly imposed arbitrary assumptions on the channel parameters to achieve theoretical tractability, thereby limiting their application to represent diverse range of propagation environments. This motivates us to find a more general fading model that incorporates multiple specular components with clusterized diffuse scattered waves, but achieves analytical tractability at the same time. To this end, we introduced the Multiple-Waves with Generalized Diffuse Scatter (MWGD) and Fluctuating Multiple-Ray (FMR) model that allow arbitrary number of specular components and assume generalized diffuse scattered model. We derive the distribution functions of the signal envelop in closed form and calculate second order statistics of the proposed fading model. Furthermore, we evaluate the performance metrics of wireless communications systems, such as the capacity, outage probability and average bit error rate. Through numerical simulations, we obtain important new insights into the link performance of the 5G communications while considering a diverse range of fading conditions and channel characteristics. 
\end{abstract}


\clearpage

\IEEEpeerreviewmaketitle

\section{Introduction}

The fifth generation of mobile technology (5G) is propositioned to achieve 1000 fold gains in capacity over the current 4G wireless networks, offer high data rates exceeding 10 Gigabits/s, provide seamless connection even to the users at the cell edge and enable ultra-reliable communications with extremely low latency of less than 1 millisecond \cite{DMCR&DCenterSamsungElectronicsCo.2015,Nokia2014}. To meet these seemingly contradicting requirements, 5G communications will require access to broad range of spectrum from below 1 GHz up to 100 GHz utilizing different properties of each band \cite{5GCM}. For example, low bands below 1 GHz will be used for cases that need wide coverage, whereas high bands above 24 GHz will be applied to cases that require extremely high throughput with very large bandwidth \cite{Zaidi2017}. In order to harness the full potential of utilizing diverse spectrum, it is crucial to develop an accurate fading model especially for frequencies above 6GHz\footnote{Fading models for frequencies below 6GHz are well developed and validated in the literature \cite{Haneda2016}.} that provides a good fit to the field measurements.

Recently, wide range of measurement campaigns had been conducted by many research organizations and consortium to statistically characterize the channel model of 5G communications, including METIS2020 \cite{METIS2020}, COST2100 \cite{Poutanen2011,Liu2012}, NYU WIRELESS \cite{MacCartney2015,Rappaport2015a,Rappaport2015,Rappaport2013,Samimi2016} and QUB \cite{Cotton2016, Yoo2017,Yoo2017a}. These measurement campaigns commonly observed the following two important characteristics:
\begin{enumerate}
	\item {\em Dominant specular scattering}: The transmission schemes that will be utilized in 5G communications are highly directive, \textit{e.g.}, millimeter wave (mmWave) or massive MIMO (multiple-input, multiple-output), which often cause the propagation environment to be dominated by the specular waves \cite{Medbo2014}. Furthermore, reflections by random objects may cause multiple dominant specular components to have significantly larger power than the diffuse waves. For example, the channel model of the mmWave transmission can often be approximated by a finite number of specular waves with no diffuse component \cite{Oestges2015,Cotton2009}. 
	\item {\em Clustering of the scattered waves}: The prevalent Rayleigh fading model assumes that the point scatterers are uniformly and homogeneously distributed. Under such assumptions, the central limit theorem (CLT) holds, so that the in-phase and quadrature components of the signal can be approximated by a Gaussian process with zero mean and equal variance. However, if there is a non-zero correlation among the clusters of the scattered waves, or there is a non-zero correlation between the in-phase and quadrature components, or the number of components are not large enough, then the CLT may not be applied \cite{Beaulieu2014,Yacoub2007a}. 
\end{enumerate}

These physical characteristics of 5G communications have been incorporated into several fading models in the literature already. In \cite{Durgin2002a}, Durgin considered a wireless channel where the received signal can be represented as a combination of multiple specular waves and a diffuse component, whose envelope follows Rayleigh distribution with scale parameter $\sigma$ as below 
\begin{equation}
    \begin{split}
        &\tilde{V} = R \mathrm{e}^{j \phi} = \sum_{i = 1}^{N} V_i \mathrm{e}^{j \phi_i} + \tilde{V}_{dif},\\ 
        &\tilde{V}_{dif} = V_d \mathrm{e}^{j \phi_d}, \quad V_d \sim \text{Rayleigh}(\sigma),
    \end{split}
    \label{q004-001}
\end{equation}
where $\tilde{V}$ denotes the received complex baseband signal with magnitude $R$ and phase $\phi$, $N$ is the number of specular waves, $\{V_i, \phi_i\}$ represents the magnitude and phase of the $i$-th specular component and $\tilde{V}_{dif}$ is the diffuse component with magnitude $V_d$ and phase $\phi_d$. 
The probability density function (PDF) of the signal envelope $R$ in (\ref{q004-001}) was originally presented in an integral expression, which was complicated to use in network performance analysis. Due to this limitation, Durgin considered the Two Wave with Diffuse Power (TWDP) model and proposed an approximated PDF form that is analytically tractable and presented in a closed form. The TWDP model was shown to provide a good fit to the field measurements in a variety of propagation scenarios \cite{Samimi2016, Rappaport1989} and widely adopted in the literature as the de-facto channel model \cite{Oh2005,Oh2007,Tan2011, Lu2012,Lu2012a,Dixit2013,Wang2014}.

In \cite{Rao2015,Rao2015a}, the authors proposed the Generalized Two-Ray Fading (GTR) model that generalizes the phase distribution of the TWDP model. Specifically, the TWDP model assume the phase of the individual component to be uniformly distributed, whereas the GTR model can incorporate arbitrary distributed phase. If the phase of the GTR model is assumed to be uniformly distributed, then the resultant model (referred as the \textit{GTR-U} model) becomes identical to the TWDP model. More practical phase models, such as the truncated phase (known as the \textit{GTR-T} model) or the Von Mises phase distribution (referred as the \textit{GTR-V} model), were introduced in \cite{Rao2015}.

The Rayleigh distributed diffuse model in (\ref{q004-001}) rely on the assumption that the CLT holds, which is certainly an approximation, because the diffuse component may be composed of relatively small number of scattered waves, or the scattered waves may be clustered with non-zero correlation among the clusters, or there might be a non-zero correlation between the in-phase and quadrature components. To reflect such channel environments, the authors in \cite{Beaulieu2014} proposed the \textit{generalized diffuse scatter model} where the envelope $V_d$ of the diffuse component follows the Nakagami-\textit{m} distribution. Since the Nakagami-\textit{m} distribution is a scaled representation of the central chi-squared distribution, which models a sum of the squares of zero-mean Gaussian random variables (RV), it can effectively characterize a diffuse component composed by the clustering of the scattered waves. Furthermore, the PDF of the signal envelope $R$ was derived in a closed form without any integral expression and the average bit error rate (BER) of an arbitrary modulation scheme was evaluated. 

The magnitudes of the specular components $\{V_1, \ldots, V_N\}$ were assumed constant in the original TWDP, GTR and generalized diffuse model. However, in practical wireless channels, the amplitude of the specular component often fluctuates due to a blockage by local objects, \textit{e.g.}, human-body shadowing  or mobile scatterers. This physical phenomenon was validated through field measurements \cite{Abdi2003,Paris2014,Cotton2014} and was incorporated in the literature on various channel environments, including the Rician shadowed fading \cite{Abdi2003} which generalizes the Rician fading model and the $\kappa$-$\mu$ shadowed fading \cite{Paris2014,Cotton2014} which generalizes the $\kappa$-$\mu$ fading model. In \cite{Romero-Jerez2016,Romero-Jerez2017}, the authors generalized the TWDP model by modulating the amplitude of the specular components, proposed the Fluctuating Two-Ray (FTR) fading model and numerically validated that the FTR model improves the fitting performance than the conventional fading models. In \cite{Beaulieu2015}, the authors completely relaxed the assumption of constant magnitudes and modeled the envelope of each components as Gaussian RVs with non-zero means. The PDF of the signal envelope in \cite{Beaulieu2015} was represented as the non-central chi-squared distribution, which includes the generalized Rician and $\kappa$-$\mu$ fading as special cases. 

Motivated by these approaches, we propose two fading models; namely, Multiple-Waves with Generalized Diffuse Scatter (MWGD) and Fluctuating Multiple-Ray (FMR) model, which are natural extension of the TWDP and FTR model, respectively. For the MWGD model, we assume arbitrary number $N$ of specular components, relax the Gaussian assumption on the diffuse component and assume that the magnitude of the diffuse component follows Nakagami-\textit{m} distribution. Based on the distribution functions, we evaluate important statistics of the proposed models, \textit{e.g.}, moments, amount of fading, moment generating function, and evaluate system performance metrics, \textit{e.g.}, ergodic capacity, outage probability and average BER.

The main contributions of this paper may be summarized as follows. 
\begin{enumerate}
	\item We consider generalized fading environments with arbitrary number of specular components, generalized diffuse scatter waves where the CLT may not hold and random fluctuation in the specular components. We derive the PDF and cumulative density function (CDF) of the received signal envelop in closed form using the confluent Hypergeometric function of $n$ variables. To improve tractability, we use Laguerre polynomial series expansion, which efficiently converts the multi-fold infinite summation inherent in the multi-variate Hypergeometric function to a single summation. 
	\item Using the series expression of the distribution functions, we derive statistics of the proposed fading model, including the moments, moment generating function, amount of fading and channel quality estimation index. Furthermore, we evaluate important performance metrics of wireless communications systems, including the capacity, outage probability and average bit error rate. Finally, we obtain the asymptotic expression of the performance metrics for low and high average signal-to-noise ratio (SNR). 
	\item We present numerical simulation results which provide useful insights into the performance of wireless communications with different channel characteristics. In particular, we observe that the MWGD and FMR fading models are inherently multimodal distribution with several different modes. However, as the number of specular component increases, the multimodality subsides and both model reduce to unimodal distributions. This information will be a paramount importance to those responsible for designing future 5G network infrastructure to ensure that adequate service can be provided. 
\end{enumerate}


The remainder of this paper is organized as follows. We describe the signal model in Section II and derive the distribution functions of the signal envelope in Section III. Using the series form expression of the distribution functions, we derive important statistics of the proposed fading models and evaluate system performance metrics in Section IV. We present numerical results in Section V and conclude the paper with some closing remarks in Section V. 



\section{Signal Model}

We consider a wireless channel where the received signal is represented as a summation of $N$ dominant specular components and a diffusely scattered wave whose magnitude follows the Nakagami-\textit{m} distribution with shape parameter $m$ and spread parameter $\Omega$ as expressed below
\begin{equation}
    \begin{split}
    \text{MWGD model:}~
    \begin{dcases}
        &\tilde{V} = R \mathrm{e}^{j \phi} = \sum_{i = 1}^{N} V_i \mathrm{e}^{j \phi_i} + \tilde{V}_{dif}, 
        \quad \tilde{V}_{dif} = V_d \mathrm{e}^{j \phi_d}\\ 
        &f_{V_d}(x) = \frac{2}{\Gamma(m)}\left( \frac{m}{\Omega} \right)^m x^{2m -1} \exp\left(-\frac{m}{\Omega} x^2\right).
    \end{dcases}
    \end{split}
    \label{q004-ext-000}
\end{equation}
We refer (\ref{q004-ext-000}) as the Multiple-Waves with Generalized Diffuse Scatter (MWGD) model, which generalize the TWDP model by extending the number of specular waves from $N = 2$ to arbitrary $N$, relaxing the Gaussian assumption and adopting the generalized diffuse scatter model. 

We also consider a variant of the MWGD model by introducing a random fluctuation $\xi$ on the specular components, which we referred as the Fluctuating Multiple-Ray (FMR) model
\begin{equation}
    \begin{split}
    \text{FMR model:}~
        &\tilde{V} = \sum_{i = 1}^{N} \sqrt{\xi} V_i \mathrm{e}^{j \phi_i} + \tilde{V}_{dif}, 
        \quad f_{\xi}(u) = \frac{{m_s}^{m_s} u^{{m_s}-1}}{\Gamma({m_s})} \mathrm{e}^{-{m_s} u},
    \end{split}
    \label{q004-ext-001}
\end{equation}
where $\xi$ is a Gamma RV with scale parameter $m_s$ and the envelope $V_d = \|\tilde{V}_{dif}\|$ of the diffuse component follows the Nakagami-\textit{m} distribution in (\ref{q004-ext-000}). Since the random fluctuation $\xi$ becomes increasingly deterministic as the parameter $m_s$ approaches infinity, the FMR model is equivalent to the MWGD model when $m_s \rightarrow \infty$. Similar to the MWGD model, the FMR model extends the FTR model \cite{Romero-Jerez2016,Romero-Jerez2017} to arbitrary number of specular components and generalized diffuse scatter model. The MWGD and FMR model are extremely versatile and include majority of the popular fading model as special cases (See Table $1$).

\begin{equation}
    \begin{split}
    \tilde{V} &= R \mathrm{e}^{j \phi} = \sum_{i = 1}^{N} V_i \mathrm{e}^{j \phi_i} + \tilde{V}_{dif}, 
        \quad \tilde{V}_{dif} = V_d \mathrm{e}^{j \phi_d}, 
    \end{split}
     \nonumber     
\end{equation}
where $\tilde{V}$ denotes the received complex baseband signal with magnitude $R$ and phase $\phi$, $N$ is the number of specular waves, $\{V_i, \phi_i\}$ represents the magnitude and phase of the $i$-th specular component and $\tilde{V}_{dif}$ is the diffuse component with magnitude $V_d$ and phase $\phi_d$. 

\begin{equation}
    \begin{split}
      f_R(r) &= r \int_{0}^{\infty} \Hypergeometric{1}{1}{1}{1}{-\frac{1}{4} \nu^2}
      \prod_{i=1}^{N} \BesselJ{0}{V_i \nu} \cdot \nu \BesselJ{0}{r \nu} \mathrm{d}\nu\\
      &= 2 \epsilon r \exp\left( -\epsilon r^2 \right) \cdot \sum_{n=0}^{\infty} C_n L_n\left( \epsilon r^2 \right),
    \end{split}
     \nonumber     
\end{equation}
where $L_n(x)$ is the Laguerre polynomial, $\gamma(s, x)$ is the lower incomplete gamma function, the positive, real-valued constant $\epsilon$ and the coefficient $C_n$ are respectively defined as below
\begin{equation}
    \begin{split}
    \epsilon = \frac{1}{\sum_{i=1}^{N} V_i^2 + \Omega}, \quad 
    C_n &= \mathbb{E}_{R}\left[ L_n(\epsilon r^2) \right].
    \end{split}
    \label{q004-IIIB-003}
\end{equation}
$\Psi_{2}\left(-;-;-\right)$ may be inappropriate for analytic studies due to its intractability. To resolve these issues, we exploit the Laguerre polynomial expansion \cite{Abdi2000,Chai2009} and derive the series form expressions of the distribution functions for both the MWGD and FMR model in the following theorem.

\section{Distribution of the Envelope $R$}

In this section, we will characterize the distribution of the received signal envelope $R$ for both MWGD and FMR model and derive the corresponding distribution of the received SNR $\gamma$
\begin{equation}
    \begin{split}
    &\gamma = \frac{E_b}{N_0} R^2 = \gamma_0 R^2, \quad 
    \bar{\gamma} = \mathbb{E}[ \gamma ] = \gamma_0 \left(\sum_{i=1}^{N} V_i^2 + \Omega \right),
    \end{split}
    \label{q004-003}
\end{equation}
where $E_b$ is the transmit symbol energy, $N_0$ is the noise power spectral density, $\gamma_0 = \frac{E_b}{N_0}$ and $\bar{\gamma}$ denotes the average received SNR. First, we will derive the PDF and CDF of the signal envelope in closed form using Hankel transform and multi-variate Hypergeometric function. Then, we will use Laguerre polynomial expansion to obtain the series form expression of the distribution functions, which is easier to use and improve analytical tractability. 


\subsection{Closed Form Expressions}

Let us consider uniformly distributed phase over $[0, 2\pi]$ for both MWGD and FMR model \cite{Durgin2002a}. Then, the PDF and CDF of the received signal envelop are derived in the following Theorems. 
\begin{thm}
  Let us consider the MWGD fading model as described in (\ref{q004-ext-000}). Then, the PDF and CDF of the received signal envelope $R$ are respectively given by 
  \begin{equation}
    \begin{split}
      f_R(r) &= r \int_{0}^{\infty} \Hypergeometric{1}{1}{m}{1}{-\frac{\Omega}{4m} \nu^2}
      \prod_{i=1}^{N} \BesselJ{0}{V_i \nu} \cdot \nu \BesselJ{0}{r \nu} \mathrm{d}\nu\\
      &\stackrel{(a_1)}{=} 
      \frac{2m}{\Omega} \cdot r \cdot \sum_{k=0}^{m-1} \binom{m-1}{k} (-1)^{k} 
      ~\Psi_{2}\left( k+1; [1]_{N+1}; \underline{\alpha}(r) \right) \quad \text{for } m \in \mathbb{Z}^{+},
    \end{split}
    \label{q004-004}
\end{equation}
  \begin{equation}
    \begin{split}
      F_R(t) &= t \int_{0}^{\infty} \Hypergeometric{1}{1}{m}{1}{-\frac{\Omega}{4m} \nu^2}
      \prod_{i=1}^{N} \BesselJ{0}{V_i \nu} \cdot \BesselJ{1}{t \nu} \mathrm{d}\nu\\
      &\stackrel{(a_2)}{=} 
      \frac{m}{\Omega} \cdot t^2 \cdot \sum_{k=0}^{m-1} \binom{m-1}{k} (-1)^{k} 
      ~\Psi_{2}\left( k+1; [1]_{N}, 2; \underline{\alpha}(t) \right)
      \quad \text{for } m \in \mathbb{Z}^{+},
    \end{split}
    \label{q004-005}
\end{equation}
where $\mathbb{Z}^{+}$ represents the set of positive integers, the equality ($a_1$) and ($a_2$) holds if the parameter $m$ takes positive integer values, $\BesselJ{n}{x}$ is the $n$-th order Bessel function of the first kind, $\Psi_{2}\left( a; b_1, \ldots b_n ; x_1, \ldots, x_n \right)$ is the Confluent Hypergeometric series of $n$ variables, which is described in (\ref{q004-eq_app_I-001}) of Appendix I, and $[x]_{n}$ and $\underline{\alpha}(x)$ represents the following sequences
  \begin{equation}
    \begin{split}
      [x]_{n} \triangleq \{ \underbrace{x, x, \ldots, x}_{n~\text{times}} \}, \quad 
      \underline{\alpha}(x) \triangleq \left\{ \frac{V_1^2 m}{\Omega}, \ldots, \frac{V_N^2 m}{\Omega}, \frac{x^2 m}{\Omega} \right\}.
    \end{split}
    \label{q004-006}
\end{equation}
\end{thm}

\begin{proof}
  See Appendix II.
\end{proof}

Based on (\ref{q004-003}), the distribution of the SNR $\gamma$ is related to the envelope $R$ as expressed below
  \begin{equation}
    \begin{split}
    f_{\gamma}(x) &= \frac{1}{2\sqrt{\gamma_0 x}} f_{R}\left( \sqrt{\frac{x}{\gamma_0}} \right), \quad 
    F_{\gamma}(x) = F_{R}\left( \sqrt{\frac{x}{\gamma_0}} \right).
    \end{split}
    \label{q004-007}
\end{equation}
Then, the following corollary can be readily obtained by substituting Theorem 1 to (\ref{q004-007}).

\begin{corollary}
  Given positive integer valued parameter $m$, the PDF and CDF of the received SNR $\gamma$ for the MWGD model are respectively given by 
  \begin{equation}
    \begin{split}
      &f_{\gamma}(x) = 
      \frac{m}{\Omega \gamma_0} \sum_{k=0}^{m-1} \binom{m-1}{k} (-1)^{k} 
      ~\Psi_{2}\left( k+1; [1]_{N+1}; \frac{V_1^2 m}{\Omega}, \ldots, \frac{V_N^2 m}{\Omega}, \frac{x m}{\Omega \gamma_0} \right), \\
      &F_{\gamma}(x) = 
      \frac{m}{\Omega \gamma_0} \cdot t \cdot \sum_{k=0}^{m-1} \binom{m-1}{k} (-1)^{k} 
      ~\Psi_{2}\left( k+1; [1]_{N}, 2; \frac{V_1^2 m}{\Omega}, \ldots, \frac{V_N^2 m}{\Omega}, \frac{x m}{\Omega \gamma_0} \right).
    \end{split}
    \label{q004-008}
\end{equation}  
\end{corollary}

As a special case of the MWGD model, let us assume $m=1$ where the received signal $\tilde{V}$ is composed of $N$ specular components and a diffuse component whose magnitude follows Rayleigh distribution. This signal model was originally introduced in \cite{Durgin2002a} and recently revisited in \cite{Chai2009} where the authors obtained the distribution of the envelope in a series form by using Laguerre polynomial expansion and referred this model as the Multiple-Waves plus Diffuse Power (MWDP) fading. By using Theorem 1 and substituting $m=1$, we can easily obtain the PDF and CDF of the MWDP fading in closed form expressions, which are given below.

\begin{corollary}
For the MWDP model, the PDF and CDF of the received signal envelope $R$ are respectively given by 
  \begin{equation}
    \begin{split}
      f_R(r) &= \frac{2 r}{\Omega} \cdot \Psi_{2}\left( 1; [1]_{N+1}; \frac{V_1^2}{\Omega}, \ldots, \frac{V_N^2}{\Omega}, \frac{r^2}{\Omega} \right),\\
      F_R(t) &= \frac{t^2}{\Omega} \cdot \Psi_{2}\left( 1; [1]_{N}, 2; \frac{V_1^2}{\Omega}, \ldots, \frac{V_N^2}{\Omega}, \frac{t^2}{\Omega} \right).
    \end{split}
    \label{q004-009}
\end{equation}  
\end{corollary}

Next, let us consider the FMR model. For a given value of $\xi$, (\ref{q004-ext-001}) essentially becomes the MWGD model where the magnitude of the $i$-th specular component is $\sqrt{\xi}V_i$. By using the total probability, the distribution functions of the FMR model are derived in the following theorem.

\begin{thm}
  Let us consider the FMR fading model as described in (\ref{q004-ext-001}). The corresponding PDF and CDF of the received signal envelope $R$ are respectively given by 
  \begin{equation}
    \begin{split}
      f_R(r) &= r \int_{0}^{\infty} \Hypergeometric{1}{1}{m}{1}{-\frac{\Omega}{4m} \nu^2}
      \Psi_{2}\left( m_s; [1]_{N}; \underline{\beta}(\nu) \right) \cdot 
      \nu \BesselJ{0}{r \nu} \mathrm{d}\nu, \\
      F_R(t) &= t \int_{0}^{\infty} \Hypergeometric{1}{1}{m}{1}{-\frac{\Omega}{4m} \nu^2}
      \Psi_{2}\left( m_s; [1]_{N}; \underline{\beta}(\nu) \right) \cdot 
      \BesselJ{1}{\nu t} \mathrm{d}\nu,      
    \end{split}
    \label{q004-010}
\end{equation}
which holds for positive real valued $m$ and $m_s$ and $\underline{\beta}(\nu)$ denotes the following sequences
  \begin{equation}
    \begin{split}
      \underline{\beta}(\nu) \triangleq \left\{ \frac{V_1^2 \nu^2}{4 m_s}, \ldots, \frac{V_N^2 \nu^2}{4 m_s} \right\}.
    \end{split}
    \label{q004-011}
\end{equation}
\end{thm}

\begin{proof}
  See Appendix III.
\end{proof}

\begin{remark}
  Although the confluent Hypergeometric series $\Psi_{2}\left(-;-;-\right)$ is defined as $n$-fold infinite summation in (\ref{q004-eq_app_I-001}), the following identities can be utilized to efficiently evaluate $\Psi_{2}\left(-;-;-\right)$
  \begin{equation}
    \begin{split}
    \Psi_2\left( k; [1]_{n}; \frac{a_1}{p}, \ldots, \frac{a_n}{p} \right)
    &\stackrel{(b_1)}{=} 
    \frac{1}{\Gamma(k)} 
    \int_{0}^{\infty} t^{k-1} \mathrm{e}^{-t} \cdot \prod_{i=1}^{n} \BesselJ{0}{2 \sqrt{{a_i t}/{p}}} \mathrm{d}t\\ 
    &\stackrel{(b_2)}{\approx} 
    \frac{1}{\Gamma(k)} 
    \sum_{k=1}^{M} w_k g\left( x_k \right) +R_M, \quad 
    g(x) \triangleq x^{k-1} \prod_{i=1}^{n} \BesselJ{0}{2 \sqrt{{a_i t}/{p}}},
    \end{split}
    \label{q004-012}
\end{equation}
where we used (\ref{eq_app_I-exo1-002}) in ($b_1$), applied (\ref{eq_app_I-002-a-ext9}) to ($b_2$) with $x_k$ and $w_k$representing the $k$-th abscissa and weight of the $M$-th order Laguerre polynomial, respectively. Since the approximation error $R_M$ converges rapidly to zero \cite{Gradshteyn1994}, the expression ($b_2$) achieves numerically accurate results.
\end{remark}

\subsection{Series Form Expressions}


In the previous subsection, we derived the distribution functions of the envelope for both the MWGD and FMR model to closed form expressions that commonly involve the confluent Hypergeometric series of multiple variables. As we described in Remark 1, $\Psi_{2}\left(-;-;-\right)$ can be numerically evaluated either via a one-dimensional integration or a finite summation. However, the numerical evaluation of the distribution functions might be time consuming if the integrand of $\Psi_{2}\left(-;-;-\right)$ has a heavy tail or it might be prone to computational errors due to the rapid fluctuation of the Bessel function and its product term $\prod_{i=1}^{N} \BesselJ{0}{V_i \nu}$. Furthermore, $\Psi_{2}\left(-;-;-\right)$ may be inappropriate for analytic studies due to its intractability. To resolve these issues, we exploit the Laguerre polynomial expansion \cite{Abdi2000,Chai2009} and derive the series form expressions of the distribution functions for both the MWGD and FMR model in the following theorem.

\begin{thm}
	The distribution functions of the signal envelope $R$ for both MWGD and FMR model can be commonly expressed in series form as below
\begin{equation}
	\begin{split}
		f_R(r) &= 2 \epsilon r \exp\left( -\epsilon r^2 \right) \cdot \sum_{n=0}^{\infty} C_n L_n\left( \epsilon r^2 \right),\\
		F_R(t) &= \sum_{n=0}^{\infty} C_n ~\sum_{k=0}^{n} \frac{(-1)^k}{k!} \binom{n}{k} \gamma\left( k+1, \epsilon t^2\right),
	\end{split}
	\label{q004-IIIB-001}
\end{equation}
where $L_n(x)$ is the Laguerre polynomial, $\gamma(s, x)$ is the lower incomplete gamma function, the positive, real-valued constant $\epsilon$
and the coefficient $C_n$ are respectively defined as below
\begin{equation}
	\begin{split}
	\epsilon = \frac{1}{\sum_{i=1}^{N} V_i^2 + \Omega}, \quad 
	C_n &= \mathbb{E}_{R}\left[ L_n(\epsilon r^2) \right] 
	= \sum_{k=0}^{n} \frac{\left( -\epsilon \right)^k}{k!} \binom{n}{k} \mathbb{E}_{R}[r^{2k}].
	\end{split}
	\label{q004-IIIB-003}
\end{equation}
Then moments $\mathbb{E}_{R}[r^{2k}]$ of the envelope $R$ can be evaluated by the following recursion formula 
\begin{equation}
	\begin{split}
	\mathbb{E}_{R}[r^{2k}] \triangleq u_{N+1}^{(2k)}, \quad 
	u_{j}^{(2k)} = \sum_{i=0}^{k} \binom{k}{i}^2 u_{j-1}^{(2i)} {\nu}_{j}^{(2k-2i)} \quad 
	\text{for } j = 2, \ldots, N+1, 
	\end{split}
	\label{q004-IIIB-004}
\end{equation}
where the initial value is $u_{1}^{(2k)} = \nu_{1}^{(2k)}$, $(x)_{n} = {\Gamma(x+n)}/{\Gamma(x)}$ is the Pochhammer symbol and
\begin{equation}
	\begin{split}
		\nu_{j}^{(2k)} = 
		\begin{dcases}
		V_j^{2k} \hfill &\text{for MWGD model}, \quad j = 1, \ldots, N,\\
		V_j^{2k} \cdot \frac{\left(m_s\right)_{k}}{m_s^{k}} \hfill &\text{for FMR model}, \quad j = 1, \ldots, N,\\
		\left( m \right)_k \left( \frac{\Omega}{m}\right)^k \hfill &\text{for } j = N+1.
		\end{dcases}
	\end{split}
	\label{q004-IIIB-005}
\end{equation}
\end{thm}

\begin{proof}
  See Appendix IV.
\end{proof}

We note that the series expression of the distributions for the MWGD and FMR model have an identical form in (\ref{q004-IIIB-001}), but the difference of each model is reflected through the coefficient $C_n$. The corresponding PDF and CDF of the received SNR $\gamma$ can also be expressed in series form by substituting Theorem 3 to (\ref{q004-007}) as given in the following corollary.

\begin{corollary}
  The PDF and CDF of the received SNR $\gamma$ are respectively given by 
  \begin{equation}
    \begin{split}
		f_{\gamma}(x) &= \frac{\epsilon}{\gamma_0} \exp\left( -\frac{\epsilon x}{\gamma_0} \right) 
		\cdot \sum_{n=0}^{\infty} C_n L_n\left( \frac{\epsilon x}{\gamma_0} \right),\\
		F_{\gamma}(x) &= \sum_{n=0}^{\infty} C_n ~\sum_{k=0}^{n} \frac{(-1)^k}{k!} \binom{n}{k} \gamma\left( k+1, \frac{\epsilon x}{\gamma_0}\right).
    \end{split}
    \label{q004-IIIB-006}
\end{equation}  
\end{corollary}

\begin{remark}
	We can verify that the integral of the PDF in (\ref{q004-IIIB-001}) over $\{0, \infty\}$ is equal to one 
	\begin{equation}
	    \begin{split}
	    \lim_{t \rightarrow \infty} \int_{0}^{t} f_R(r) \mathrm{d}r &= 
	    \sum_{n=0}^{\infty} C_n ~\sum_{k=0}^{n} \frac{(-1)^k}{k!} \binom{n}{k} \Gamma\left( k+1 \right)
	    = C_0 = 1,
	    \end{split}
	    \label{q004-IIIB-007}
	\end{equation}  	
	where the first equality follows by applying the CDF form in (\ref{q004-IIIB-001}) and $\lim_{t \rightarrow \infty} \gamma\left( k, t\right) = \Gamma(k)$, the second equality is due to the following binomial equality 
	\begin{equation}
		\begin{split}
			\sum_{k=0}^{n} (-1)^k \binom{n}{k} = 
			\begin{dcases}
			1 \hfill &\text{for }n = 0,\\
			0 \hfill &\text{otherwise},
			\end{dcases}
		\end{split}
		\label{q004-IIIB-008}
	\end{equation}	
	and the last equality follows by using $L_0(x) = 1$ and $C_0 = \mathbb{E}_{R}\left[ L_0(\epsilon r^2)\right] = 1$. Hence, the series expression of the PDF in (\ref{q004-IIIB-001}) is a valid PDF. 
\end{remark}

\begin{remark}
	For numerical evaluation, we use the truncated form of (\ref{q004-IIIB-001}) where we evaluate only up to the $M$-th term and ignore the higher order terms as follows
	\begin{equation}
		\begin{split}
			{\widehat{f}_R}(r) &= 2 \epsilon r \exp\left( -\epsilon r^2 \right) \cdot \sum_{n=0}^{M} C_n L_n\left( \epsilon r^2 \right).
		\end{split}
		\label{q004-IIIB-009}
	\end{equation}
The approximation error caused by the truncation is upper-bounded \cite{Chai2009} as follows
	\begin{equation}
		\begin{split}
		e_M = \int_{0}^{\infty} \left( f_R(r) - {\widehat{f}_R}(r) \right)^2 \mathrm{d}r \leq 
		\left( \epsilon \pi \right)^{\frac{1}{2}} \left[\sum_{n=M+1}^{\infty} |C_n|\right]^{2},
		\end{split}
		\label{q004-IIIB-010}
	\end{equation}
where $f_R(r)$ is the exact PDF and ${\widehat{f}_R}(r)$ is the truncated PDF. The upper bound converges to zero as $M$ goes to infinity, thereby the approximation error converges to zero as well.
	\begin{equation}
		\begin{split}
		\lim_{M \rightarrow \infty} \left[\sum_{n=M+1}^{\infty} |C_n|\right]^{2} = 0 ~\Leftrightarrow~
		\lim_{M \rightarrow \infty} e_M = 0.
		\end{split}
		\label{q004-IIIB-011}
	\end{equation}
Hence, (\ref{q004-IIIB-009}) provides a tight approximation to the exact PDF of the signal envelope $R$.
\end{remark}

\begin{remark}
	Theorem 3 achieves a series representation of the confluent Hypergeometric series $\Psi_{2}\left(-;-;-\right)$. For example, the series expansion of the hypergeometric series in (\ref{q004-009}) can be easily obtained by comparing (\ref{q004-009}) to (\ref{q004-IIIB-001}), which achieves the following identity
	\begin{equation}
		\begin{split}
		\Psi_{2}\left( 1; [1]_{N+1}; \frac{V_1^2}{\Omega}, \ldots, \frac{V_N^2}{\Omega}, \frac{r^2}{\Omega} \right) 
		 &= \frac{\epsilon}{\Omega}\exp\left( -\epsilon r^2 \right) \cdot \sum_{n=0}^{\infty} C_n L_n\left( \epsilon r^2 \right).\\
		\end{split}
		\label{q004-IIIB-012}
	\end{equation}
	The functional identity in (\ref{q004-IIIB-012}) essentially reduces $N+1$ fold infinite summation in (\ref{q004-eq_app_I-001}) to a single summation. To the best of the author's knowledge, this is a new result. 
\end{remark}



\section{Performance Analysis of Wireless Communications System}

In this section, we use the series form of the distribution functions to derive important statistics of the proposed MWGD and FMR model and evaluate system performance metrics. 

\subsection{Statistics of the Proposed Fading Models}

\subsubsection{Moments}
The $l$-th order moment of the SNR $\gamma$ can be obtained as follows
\begin{equation}
	\begin{split}
	\mathbb{E}[\gamma^l] &= \frac{1}{\bar{\gamma}} \sum_{n=0}^{\infty} C_n
	\int_{0}^{\infty} x^l \exp\left(-\frac{x}{\bar{\gamma}}\right)
		L_n\left( \frac{x}{\bar{\gamma}} \right) \mathrm{d}x
		= \bar{\gamma}^l \sum_{n=0}^{\infty} C_n 
	\int_{0}^{\infty} t^l \mathrm{e}^{-t} L_n\left(t\right) \mathrm{d}t\\
	&= \bar{\gamma}^l \sum_{n=0}^{\infty} C_n \cdot 
	\Hypergeometric{2}{1}{-n, l+1}{1}{1}
	= \bar{\gamma}^l \sum_{n=0}^{l} \left(-1\right)^n \binom{l}{n} C_n,
	\quad \text{for } l \in \mathbb{Z}^{+},
	\end{split}
	\label{q004-IV-001}
\end{equation}
where we used (\ref{q004-IIIB-006}) in the first equality, replaced the term $\frac{\gamma_0}{\epsilon}$ to $\bar{\gamma} = \gamma_0 \left(\sum_{i=1}^{N} V_i^2 + \Omega \right)$ based on (\ref{q004-IIIB-003}), applied a change of variable, \textit{i.e.}, $\frac{x}{\bar{\gamma}} = t$, and \cite[7.414-7]{Gradshteyn1994} in the third equality. 
In the last equality, we used the following functional identity for integer valued index $l$ and $n$
\begin{equation}
	\begin{split}
		\Hypergeometric{2}{1}{-n, l+1}{1}{1} = 
		\begin{dcases}
		\binom{l}{n}(-1)^n \hfill &\text{for } 0 \leq n \leq l,\\
		0 \hfill &\text{for }n > l.
		\end{dcases}
	\end{split}
	\label{q004-IV-002}
\end{equation}
For example, the first and second moments of the SNR are given by 
\begin{equation}
	\begin{split}
	\mathbb{E}[\gamma] &= \bar{\gamma} \left( C_0 - C_1 \right),\quad 
	\mathbb{E}[\gamma^2] = \bar{\gamma}^2 \left( C_0 - 2 C_1 +C_2 \right),
	\end{split}
	\label{q004-IV-003}
\end{equation}
which can be further simplified by directly evaluating the coefficient $C_0$ and $C_1$ as follows
\begin{equation}
	\begin{split}
	&C_0 = \mathbb{E}_{R}\left[ L_0(\epsilon r^2)\right] = 1 \quad \text{since $L_0(x) = 1$ for arbitrary $x$},\\	
	&C_1 = 1 - \epsilon \mathbb{E}_{R}\left[ r^2\right] = 1 - \epsilon \left(\sum_{i=1}^{N} V_i^2 + \Omega\right) = 0
	\quad \text{due to (\ref{q004-IIIB-003})}.
	\end{split}
	\label{q004-IV-004}
\end{equation}
Then, (\ref{q004-IV-003}) can be written as follows
\begin{equation}
	\begin{split}
	\mathbb{E}[\gamma] &= \bar{\gamma},\quad 
	\mathbb{E}[\gamma^2] = \bar{\gamma}^2 \left( 1 + C_2 \right), \quad 
	\mathrm{Var}[\gamma] = \mathbb{E}[\gamma^2] - \mathbb{E}[\gamma]^2 = C_2 \cdot \bar{\gamma}^2.
	\end{split}
	\label{q004-IV-005}
\end{equation}

In the following subsection, we will consider two important metrics to assess the performance of the wireless communications systems over fading channels; namely, Amount of Fading (AF) \cite{Simon2005} and Channel Quality Estimation Index (CQEI) \cite{Lioumpas2007}, and use (\ref{q004-IV-005}) to simplify the expressions. 

\subsubsection{Amount of Fading (AF) and Channel Quality Estimation Index (CQEI)}

The AF is defined as the ratio of the variance of the SNR to the squared average SNR \cite{Simon2004}, whereas the CQEI is defined as the ratio of the variance of the SNR to the cubed average SNR \cite{Lioumpas2007}. Both metrics are often used to quantify the severity of fading experienced during transmission. Since these metrics takes into account the higher order moments of the SNR, the AF and CQEI can efficiently quantify the error performance of the wireless communications systems, especially in diversity systems. By using (\ref{q004-IV-005}), the AF and CQEI for the proposed fading model can be written as below
\begin{equation}
	\begin{split}
	&\mathrm{AF} = \frac{\mathrm{Var}[\gamma]}{\mathbb{E}[\gamma]^2} 
	= \frac{\mathbb{E}[\gamma^2] - \mathbb{E}[\gamma]^2}{\mathbb{E}[\gamma]^2} = C_2,\quad 
	\mathrm{CQEI} = \frac{\mathrm{Var}[\gamma]}{\mathbb{E}[\gamma]^3}  = \frac{\mathrm{AF}}{\mathbb{E}[\gamma]}
	= \frac{C_2}{\bar{\gamma}}.
	\end{split}
	\label{q004-IV-006}
\end{equation}

\subsubsection{Moment Generating Function (MGF)}

The MGF of the SNR is defined as follows \cite{DiRenzo2010}
\begin{equation}
	\begin{split}
	\mathcal{M}_{\gamma}(s) = \mathbb{E}[\mathrm{e}^{-s \gamma}] = \int_{0}^{\infty} \mathrm{e}^{-s x} f_{\gamma}(x) \mathrm{d}x.
	\end{split}
	\label{q004-IV-007}
\end{equation}
By substituting the series expression of the PDF from (\ref{q004-IIIB-006}) and \cite[7.414]{Gradshteyn1994} to (\ref{q004-IV-007}), the MGF of the SNR can be readily obtained as follows
\begin{equation}
	\begin{split}
	\mathcal{M}_{\gamma}(s) = \frac{1}{\left( 1 + s \bar{\gamma} \right)} \cdot \sum_{n=0}^{\infty} C_n \left( 
	\frac{s \bar{\gamma}}{1 + s \bar{\gamma}}
	\right)^n.
	\end{split}
	\label{q004-IV-008}
\end{equation}
An alternative expression of the MGF can be derived by using the higher order moments of the SNR given in (\ref{q004-IV-001}) and the Taylor series expression of the exponential function as given below
\begin{equation}
	\begin{split}
	\mathcal{M}_{\gamma}(s) = \sum_{n=0}^{\infty} \frac{\left( -s \right)^n}{n!}  \mathbb{E}[\gamma^n] 
	= \sum_{n=0}^{\infty} \frac{\left( - s \bar{\gamma} \right)^n}{n!} \left(
	\sum_{k=0}^{n} \left(-1\right)^k \binom{n}{k} C_k \right),
	\end{split}
	\label{q004-IV-009}
\end{equation}
which is beneficial for analytical evaluation of the system performance metrics due to its tractability. In the following subsection, we will use (\ref{q004-IV-009}) to calculate the ergodic capacity.

\subsection{System Performance Metrics}

\subsubsection{Ergodic Capacity}

The capacity of the proposed fading models can be obtained as follows 
\begin{equation}
	\begin{split}
		\mathrm{C} &= \mathbb{E}_{\gamma}\left[ \log_2(1+\gamma)\right] \quad (\text{bits/Sec/Hz}) \\
		&= \log_2 \mathrm{e} \sum_{n=0}^{\infty} C_n \sum_{k=0}^{n} \frac{\left(-1\right)^k}{k!} \binom{n}{k}
		\int_{0}^{\infty} t^k \mathrm{e}^{-t} \ln(1+ \bar{\gamma} t) \mathrm{d}t\\
		&= \log_2 \mathrm{e} \sum_{n=0}^{\infty} C_n \sum_{k=0}^{n} \frac{\left(-1\right)^k}{k!} \binom{n}{k}
		\sum_{l=0}^{k} \frac{k!}{(k-l)!}
		\left[ \frac{(-1)^{k-l-1}}{\bar{\gamma}^{k-l}} \mathrm{e}^{\frac{1}{\bar{\gamma}}}
		\mathrm{E}_i\left( -\frac{1}{\bar{\gamma}}\right) + \sum_{m=1}^{k-l} \frac{(m-1)!}{\left( -{\bar{\gamma}}\right)^{k-l-m}}
		\right],
	\end{split}
	\label{q004-IV-010}
\end{equation}
where we applied (\ref{q004-IIIB-006}), used a change of variable, \textit{i.e.},  $t \leftarrow \frac{x}{\bar{\gamma}}$, with (\ref{eq_app_I-001-a}) in the second equality. The last equality follows by \cite[4.337.5]{Gradshteyn1994}, where $\mathrm{E}_i(x)$ represents the exponential integral.

An alternative expression of the capacity can be obtained by using the MGF as follows
\begin{equation}
	\begin{split}
		\mathrm{C} &= \log_2 \mathrm{e} \cdot \int_{0}^{\infty} \mathrm{E}_i\left( -s \right) \cdot \frac{\partial \mathcal{M}_{\gamma}(s)}{\partial s} \mathrm{d}s
		= -\log_2 \mathrm{e} \cdot \int_{0}^{\infty} \mathrm{E}_i\left( -s \right) 
		\left(\sum_{n=0}^{\infty} \frac{\left( -s \right)^n}{n!}  \mathbb{E}[\gamma^{n+1}] \right) \mathrm{d}s\\
		&= \log_2 \mathrm{e} \cdot \sum_{n=0}^{\infty} \frac{\left( -1 \right)^n}{n!}  \mathbb{E}[\gamma^{n+1}]
		\int_{0}^{\infty} \mathrm{E}_1\left( s \right) s^n \mathrm{d}s
		= \log_2 \mathrm{e} \cdot \sum_{n=0}^{\infty} \frac{\left( -1 \right)^n}{n+1}  \mathbb{E}[\gamma^{n+1}]\\
		&= \log_2 \mathrm{e} \cdot \sum_{n=0}^{\infty} \frac{\left( -1 \right)^n \bar{\gamma}^{n+1}}{n+1} 
		 \left[\sum_{k=0}^{n+1} \left(-1\right)^k \binom{n+1}{k} C_k \right],
	\end{split}
	\label{q004-IV-011}
\end{equation}
where the first equality follows by \cite[eq.7]{DiRenzo2010}, the second equality is achieved by shifting the index $n$ in (\ref{q004-IV-009}) to incorporate the first order differentiation over $s$, the third equality is obtained by the property of the exponential integral, \textit{i.e.}, $\mathrm{E}_1(x) = -\mathrm{E}_i(-x)$ for any positive $x$, the fourth equality is due to \cite[6.223]{Gradshteyn1994} and the last equality follows by substituting (\ref{q004-IV-001}). Based on (\ref{q004-IV-011}), the asymptotic behavior of the capacity at low average SNR can be observed by ignoring the higher order terms $\bar{\gamma}^{n+1}$ ($n \geq 1$) and keeping only the dominant term ($n = 0$) in (\ref{q004-IV-011}) as follows
\begin{equation}
	\begin{split}
		\lim_{\bar{\gamma} \rightarrow 0}\mathrm{C} &= \log_2 \mathrm{e} \cdot \bar{\gamma},
	\end{split}
	\label{q004-IV-012}
\end{equation}
which accords with the expression derived in \cite[eq.12]{DiRenzo2010}. 

\subsubsection{Outage Probability and Outage Capacity}
The outage probability (or outage capacity) is defined as the probability that the instantaneous SNR (or instantaneous capacity) falls below a predefined threshold $\gamma_{th}$ (or $R_{th}$) \cite{Goldsmith2005}. By using (\ref{q004-IIIB-006}), the outage probability (or outage capacity) of the proposed fading model can be expressed as follows
  \begin{equation}
    \begin{split}
    \begin{dcases}
    	\text{Outage Probability: }	\quad &\mathbb{P}\left( \gamma \leq \gamma_{th} \right) = F_{\gamma}(\gamma_{th})\\
    	 &\qquad = \sum_{n=0}^{\infty} C_n ~\sum_{k=0}^{n} \frac{(-1)^k}{k!} \binom{n}{k} \gamma\left( k+1, \frac{\gamma_{th}}{\bar{\gamma}}\right),\\
    	\text{Outage Capacity: }	\quad &\mathbb{P}\left( \log_2(1+\gamma) \leq R_{th} \right) = F_{\gamma}\left(2^{R_{th}-1}\right)\\ 
    	&\qquad = \sum_{n=0}^{\infty} C_n ~\sum_{k=0}^{n} \frac{(-1)^k}{k!} \binom{n}{k} \gamma\left( k+1, \frac{2^{R_{th}-1}}{\bar{\gamma}}\right).
    \end{dcases}
    \end{split}
    \label{q004-IV-013}
\end{equation}  
By using (\ref{eq_app_I-exo1-003}), the asymptotic behavior of the outage probability (or outage capacity) at high average SNR can be obtained as follows
  \begin{equation}
    \begin{split}
    \bar{\gamma} \rightarrow \infty \Rightarrow
    \begin{dcases}
    	\text{Outage Probability: }	\quad 
    	&\sum_{n=0}^{\infty} C_n ~\sum_{k=0}^{n} \frac{(-1)^k}{(k+1)!} \binom{n}{k} \left(\frac{\gamma_{th}}{\bar{\gamma}}\right)^{k+1},\\
    	\text{Outage Capacity: }	\quad 
    	&\sum_{n=0}^{\infty} C_n ~\sum_{k=0}^{n} \frac{(-1)^k}{(k+1)!} \binom{n}{k} 
    	\left(\frac{2^{R_{th}-1}}{\bar{\gamma}}\right)^{k+1}.
    \end{dcases}
    \end{split}
    \label{q004-IV-014}
\end{equation}

\subsubsection{Average Bit Error Rate (BER)}

In \cite{Lopez-Martinez2010}, the generalized BER analysis for an arbitrary modulation scheme was introduced, where the instantaneous BER of a coherent detection and its first order differentiation at a given SNR $\gamma$ were expressed as
  \begin{equation}
    \begin{split}
    \mathbb{P}_{E}(\gamma) & = \sum_{r=1}^{R} \alpha_r Q\left( \sqrt{\beta_r \gamma }\right), \quad 
    \frac{\partial \mathbb{P}_{E}(\gamma)}{\partial \gamma} = 
    -\sum_{r=1}^{R} \alpha_r \sqrt{\frac{\beta_r}{8\pi \gamma}} \exp\left( -\frac{\beta_r \gamma}{2}\right),
    \end{split}
    \label{q004-IV-015}
\end{equation}  
where $Q(-)$ is the Gauss Q-function and $\{ \alpha_r, \beta_r \}_{r=1}^{R}$ are the modulation-dependent constants. By using (\ref{q004-IV-015}) and (\ref{q004-IIIB-006}), the average BER of the proposed fading models can be obtained as follows
  \begin{equation}
    \begin{split}
    \mathbb{E}_{\gamma}\left[\mathbb{P}_{E}(\gamma)\right] &= 
    \int_{0}^{\infty} \mathbb{P}_{E}(x) f_{\gamma}(x)\mathrm{d}x 
    = -\int_{0}^{\infty} \frac{\partial \mathbb{P}_{E}(x)}{\partial x} \cdot F_{\gamma}(x)\mathrm{d}x \\
    &= \sum_{r=1}^{R} \alpha_r \sqrt{\frac{\beta_r}{8\pi}} 
    \sum_{n=0}^{\infty} C_n \sum_{k=0}^{n} \frac{(-1)^k}{k!} \binom{n}{k} \cdot 
    \int_{0}^{\infty} x^{-\frac{1}{2}} \exp\left( -\frac{\beta_r x}{2}\right)
    \gamma\left( k+1, \frac{x}{\bar{\gamma}}\right) \mathrm{d}x\\
    &= \sum_{r=1}^{R} \alpha_r \sqrt{\frac{\beta_r \bar{\gamma}}{8\pi}} 
    \sum_{n=0}^{\infty} C_n \sum_{k=0}^{n} \frac{(-1)^k}{\left( k + \frac{3}{2}\right)_{\frac{1}{2}}} \binom{n}{k} 
    \cdot 
    \frac{\Hypergeometric{2}{1}{1, k+\frac{3}{2}}{k+1}{\frac{1}{1 + \frac{\beta_r \bar{\gamma}}{2}}}}{\left( 1 + \frac{\beta_r \bar{\gamma}}{2}\right)^{k+\frac{3}{2}}},
    \end{split}
    \label{q004-IV-016}
\end{equation} 
where we used integration by parts in the second equality, substituted (\ref{q004-IIIB-006}) in the third equality and applied \cite[6.455]{Gradshteyn1994} in the last equality. 



\section{Numerical Evaluation}

In this section, we present numerical evaluations of the MWGD and FMR fading models and compare the link performance over various fading parameters. All of the simulation were carried out using  MATLAB. Monte Carlo simulations have been carried out in order to validate the theoretical results, but they are not represented in these figures as the simulated values are identical to the theoretical results. We adopt a power parameter $K^{(N)}$ that is similar to the Rice factor and defined as the ratio of the power of the specular components over the power of the diffuse waves as follows \cite{Yu2007}
  \begin{equation}
    \begin{split}
    K^{(N)} = \frac{\sum_{i=1}^{N} V_i^2}{\Omega}.
    \end{split}
    \label{q004-V-001}
\end{equation}  

\begin{figure}[!t]
	\centering
	\includegraphics[width=\textwidth]{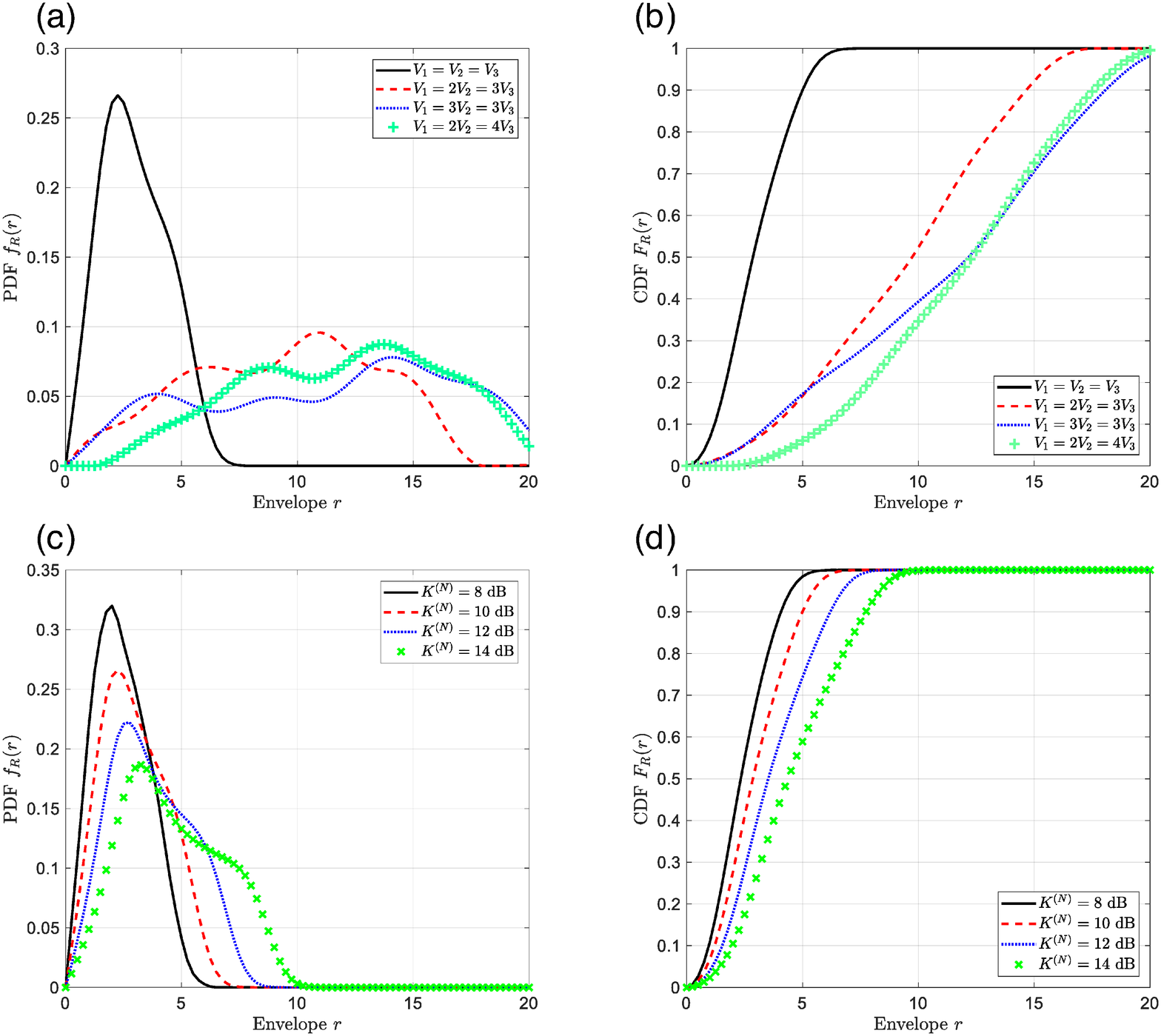}
	\caption{The PDF and CDF of the received signal envelope for MWGD fading model with $N = 3$ waves; (a) PDF for $K^{(N)} = 10$ dB, (b) CDF for $K^{(N)} = 10$ dB, (c) PDF for $V_1 = V_2 = V_3$, (d) CDF for $V_1 = V_2 = V_3$ case, respectively.}
	\label{fig:draft0_fig1}	
\end{figure}

\figref{fig:draft0_fig1} shows the PDF and CDF of the received signal envelope $R$ of the MWGD model for three specular waves ($N = 3$) scenario. In Figs. \ref{fig:draft0_fig1}(a)-(b), we fixed the power ratio at $K^{(N)} = 10$ dB, assumed $\Omega = m = 1$ for the diffuse waves and compared the link performance over various combinations of the specular components $(V_1, V_2, V_3)$. We observe that the PDF curves in \figref{fig:draft0_fig1}(a) follows an unimodal distribution for the case of $V_1 = V_2 = V_3$, whereas the other cases commonly achieve multi-modal distribution with two or three different modes. The PDF curves spread wider as the multimodality gets stronger, which cause the corresponding CDF curves in \figref{fig:draft0_fig1}(b) to be shifted toward a higher value of $R$, \textit{i.e.}, \textit{right-shifted}. In Figs. \ref{fig:draft0_fig1}(c)-(d), we assumed $V_1 = V_2 = V_3$, $\Omega = m = 1$ and compared the link performance over various power ratio $K^{(N)}$. We observe that the multimodality in the PDF curves becomes more evident as the ratio $K^{(N)}$ gets larger, causing the CDF curves to be right-shifted to a higher value of $R$. 

\begin{figure}[!t]
	\centering
	\includegraphics[width=\textwidth]{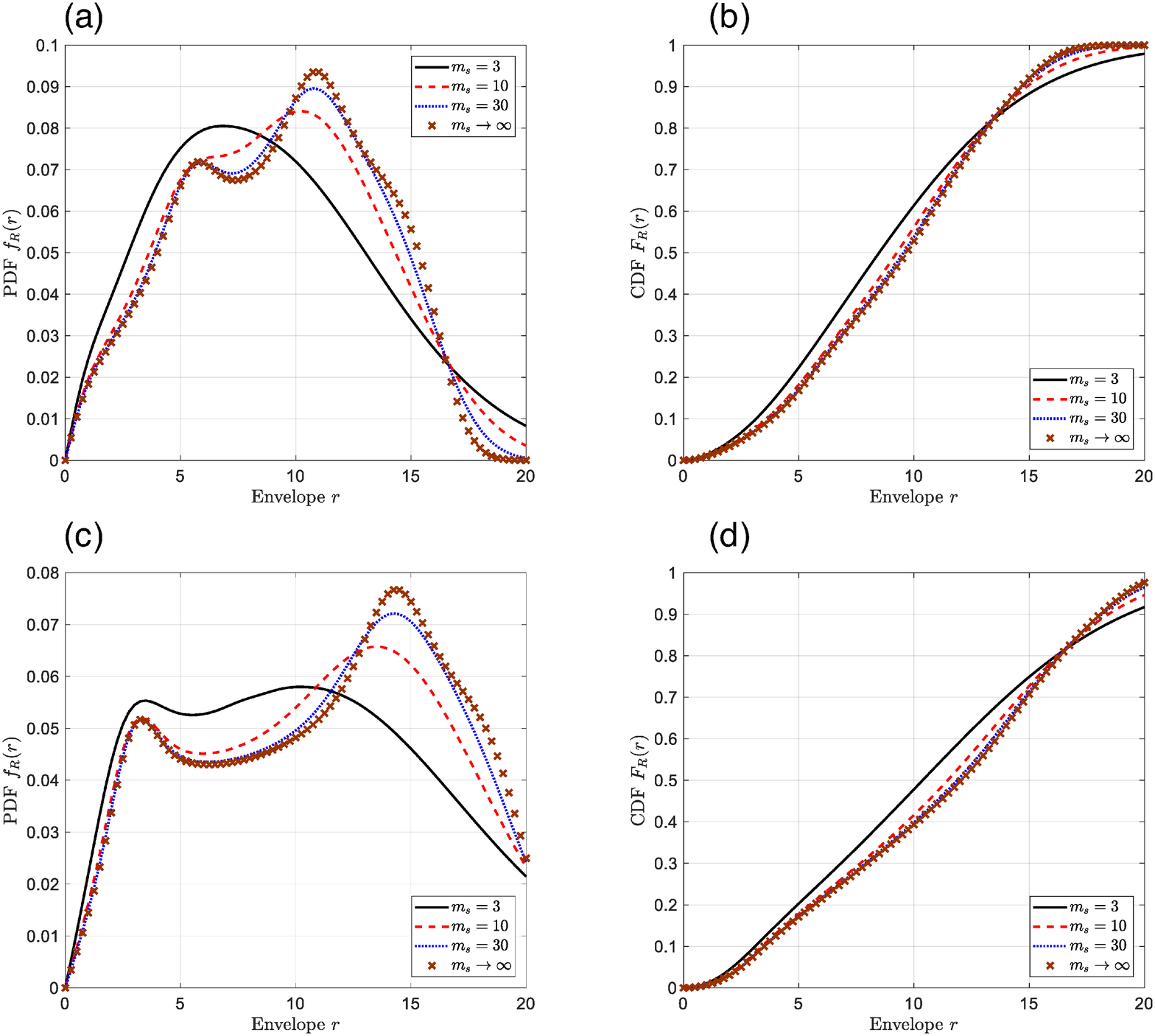}
	\caption{The PDF and CDF of the received signal envelope for FMR fading model with $N = 3$ waves; (a) PDF for $V_1 = 2 V_2 = 3 V_3$, (b) CDF for $V_1 = 2 V_2 = 3 V_3$, (c) PDF for $V_1 = 3 V_2 = 3 V_3$, (d) CDF for $V_1 = 3 V_2 = 3 V_3$ case, respectively.}
	\label{fig:draft0_fig2}	
\end{figure}

\figref{fig:draft0_fig2} shows the PDF and CDF of the received signal envelope $R$ of the FMR model for three specular waves ($N = 3$) scenario, where we assumed $K^{(N)} = 10$ dB and $\Omega = m = 1$. We considered two different combinations of the specular components $(V_1, V_2, V_3)$; In Figs. \ref{fig:draft0_fig2}(a)-(b), we considered the case of $V_1 = 2 V_2 = 3 V_3$, whereas in Figs. \ref{fig:draft0_fig2}(c)-(d), the case of $V_1 = 3 V_2 = 3 V_3$ is assumed. We observe that the multimodality in the distribution subsides as the $m_s$ parameter decreases and the PDF curves reduces to an unimodal distribution, \textit{e.g.}, $m_s = 3$ case. As the parameter $m_s$ approaches infinity, the random fluctuation in the specular components becomes deterministic and the FMR model reduces to the MWGD model.

\begin{figure}[!t]
	\centering
	\includegraphics[width=\textwidth]{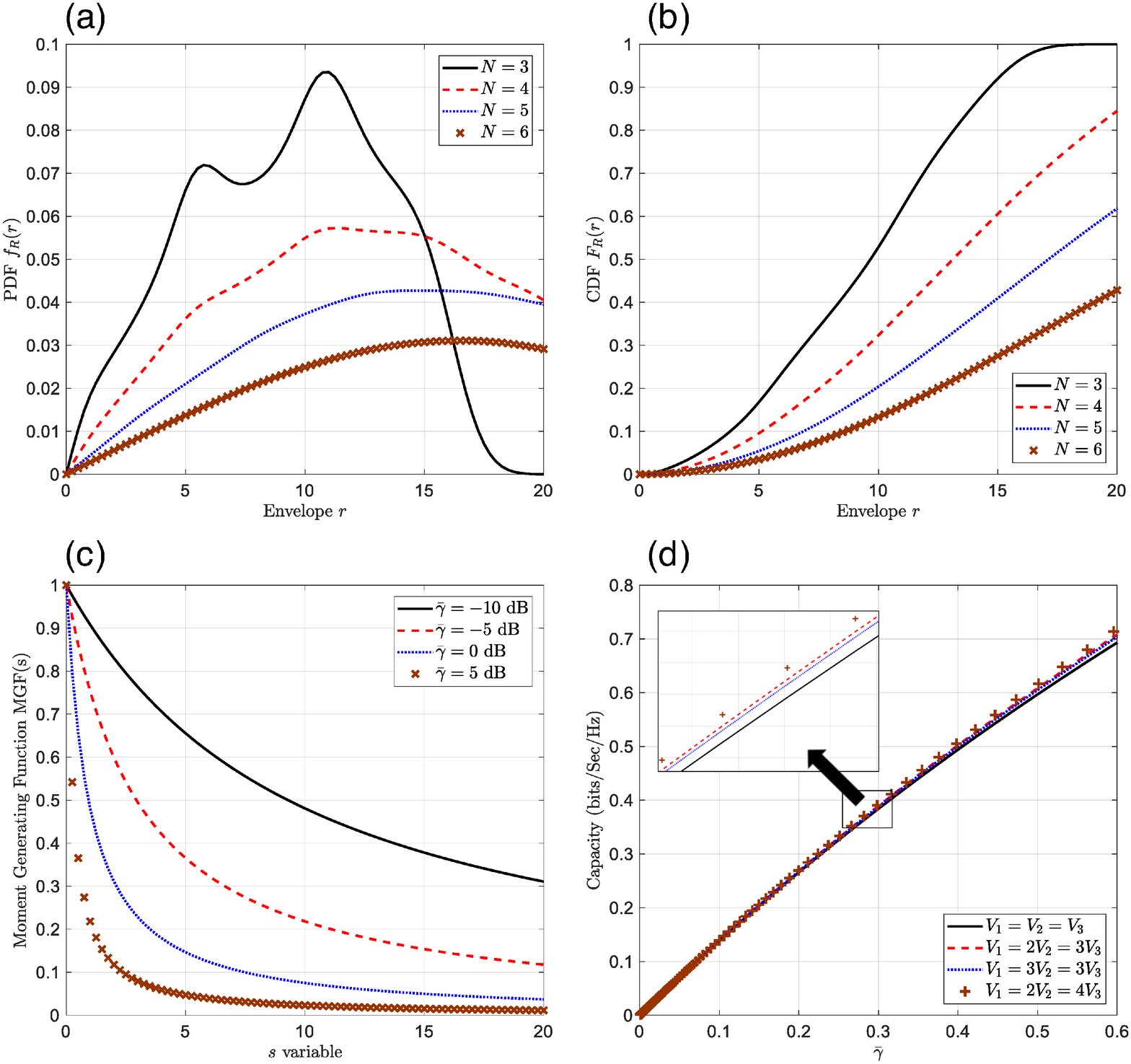}
	\caption{(a) The PDF of the MWGD model, (b) CDF of the MWGD model, (c) Moment generating function versus $s$ variable for various average SNR values, (d) Capacity versus average SNR for different combinations of the specular components.}
	\label{fig:draft0_fig3}	
\end{figure}

Figs. \ref{fig:draft0_fig3}(a)-(b) compares the PDF and CDF of the received signal envelope $R$ of the MWGD model over different number of specular components $N$, where we assumed $K^{(N)} = 10$ dB, $\Omega = m = 1$ and $V_1 = n \cdot V_n$ for $1 \leq n \leq N$. We observe that the multimodality is inherent in the MWGD model for small valued $N$. However, as the number of specular components increases, \textit{e.g.}, $N \geq 6$, the CLT holds and the PDF curves in \figref{fig:draft0_fig3}(a) follows an unimodal distribution. 

\figref{fig:draft0_fig3}(c) shows the moment generating function $\mathcal{M}_{\gamma}(s) = \mathbb{E}[\mathrm{e}^{-\gamma s}]$ of the MWGD model for different average SNR values with parameters $N=3$, $K^{(N)} = 10$ dB, $\Omega = m = 1$ and $V_1 = V_2 = V_3$. We note that the MGF of the SNR is an exponentially decaying function of $s$ and the decaying rate increases for a higher average SNR $\bar{\gamma}$. \figref{fig:draft0_fig3}(d) shows the ergodic capacity of the MWGD model for different combinations of the spectral components $(V_1, V_2, V_3)$ using parameters $N=3$, $K^{(N)} = 10$ dB and $\Omega = m = 1$. Based on the definition of the ratio $K^{(N)}$ in (\ref{q004-V-001}), the average received SNR in (\ref{q004-004}) can be expressed as $\bar{\gamma} = \gamma_0 \cdot \Omega \left( K^{(N)} + 1\right)$. Since the parameters $\gamma_0$, $\Omega$ and $K^{(N)}$ are constant in \figref{fig:draft0_fig3}(d), the capacity for different combinations of $(V_1, V_2, V_3)$ achieves similar values, which complies to the asymptotic behavior in (\ref{q004-IV-012}) for low average SNR.


\section{Conclusion}

In this paper, we introduced the MWGD and FMR fading models, which naturally extend the well-known TWDP and FTR model by generalizing the number of specular components from two to arbitrary $N$ and assuming generalized diffuse scattered model. We derived the PDF and CDF of the received signal envelop in closed form using the confluent Hypergeometric function of $n$ variables. To improve tractability, we applied the Laguerre polynomial expansion, which efficiently converts the multi-fold infinite summation inherent in the multi-variate Hypergeometric function to a single summation. Using the series form of the distribution functions, we derived important statistics of the proposed fading model, including the moments, moment generating function, amount of fading and channel quality estimation index. Furthermore, we evaluated the performance metrics of wireless communications systems, such as the capacity, outage probability and average bit error rate. Finally, we investigated the performance over a range of channel parameters using numerical simulations and observed that the MWGD and FMR fading models are inherently multimodal distribution. However, as the number of specular component increases, the multimodality subsides and both model follow unimodal distributions. 

\begin{equation}
	\begin{split}
	&\mathbf{P}\left(\text{SINR} > \zeta\right) = a_0^{-1} + \sum_{l=0}^{\infty} \sum_{n=1}^{l+\mu-1} \frac{w_l (-1)^n}{n!}  C_n,\\
	&a_k = \frac{\left( -\zeta \right)^k\delta}{k-\delta} \sum_{l=0}^{\infty} w_l (l+\mu)_{k} \cdot \sum_{j=1}^{4} \eta_j(\bar{\nu}_j)^k \pFq{2}{1}{l + \mu + k, k-\delta}{k+1-\delta}{-\bar{\nu}_j\zeta}\\
	&C_n = \sum_{k=1}^{n} \frac{B_{n,k}(a_1,\ldots,a_{n-k+1}) \Gamma\left( k+1 \right)}{a_0^{k+1}}, \quad 
	w_l  = \frac{(m)_l}{l!}\left( 1 - \frac{\theta_1}{\theta_2} \right)^l \left( \frac{\theta_1}{\theta_2} \right)^m
	\end{split}
	\label{eq:main_cp}
\end{equation}


\section*{Appendix I}

In this appendix, we summarize the operational equalities of the special functions, which are used in this paper\footnote{Most of the expressions in Appendix I were introduced in \cite{Gradshteyn1994}, except for (\ref{eq_app_I-002-a-ext5}) and (\ref{eq_app_I-002-a-ext6}), which were proved in \cite{Saad2003}.}. First, the generalized Laguerre polynomial of degree $n$ and order $\beta$ has the following functional identities
        \begin{equation}
          \begin{split}
          L_{n}^{\beta}(t) &= \sum_{i = 0}^{n} (-1)^i \binom{n+\beta}{n-i} \frac{t^i}{i!}, \quad 
          L_{n}^{0}(t) = L_{n}(t), 
          \end{split}
          \label{eq_app_I-001-a}
      \end{equation}
        \begin{equation}
          \begin{split}
      t^{\beta} \exp\left( -t \right) L_{n}^{\beta}(t) \mathrm{d}t &= \frac{1}{n}
      \mathrm{d}\left[ t^{\beta+1} \exp\left( -t \right) L_{n-1}^{\beta+1}(t) \right],\\
      \int_{0}^{t} 2 \beta r \mathrm{e}^{-\beta r^2} L_{m}\left( \beta r^2\right) \mathrm{d}r &= 
      \sum_{k=0}^{m}\frac{(-1)^k}{k!}\binom{m}{k} \gamma\left(k+1, \beta t^2 \right),
          \end{split}
          \label{eq_app_I-001-b}
      \end{equation}
where $\gamma(s, x) = \int_{0}^{x} t^{s-1} \mathrm{e}^{-t} \mathrm{d}t$ represents the lower incomplete gamma function.
The following identities holds for the lower incomplete gamma function with arbitrary positive real constant $s$
        \begin{equation}
          \begin{split}
    \frac{\gamma(s, x)}{\Gamma(s)} &= \sum_{n = 0}^{\infty}
    \frac{x^{s+n} \mathrm{e}^{-x}}{\Gamma(s + n + 1)},\quad 
    \gamma(s, x) = s^{-1} x^s \mathrm{e}^{-x} \pFq{1}{1}{1}{1+s}{x},
          \end{split}
    \label{eq_app_I-002-a-ext4}
      \end{equation}
        \begin{equation}
          \begin{split}
          \lim_{x \rightarrow 0} \gamma(s, x) = \frac{x^s}{s}.
          \end{split}
    \label{eq_app_I-exo1-003}
      \end{equation}

The following properties of hypergeometric function hold for real constants $a, b$ and $c$
        \begin{equation}
          \begin{split}
          \pFq{1}{1}{a}{b}{t} &= \mathrm{e}^{t} \pFq{1}{1}{b-a}{b}{-t},\\
          \pFq{2}{1}{a, b}{c}{z} &= (1-z)^{-a} \pFq{2}{1}{a, c-b}{c}{\frac{z}{z-1}},
          \end{split}
          \label{eq_app_I-002-a-ext1}
      \end{equation}
        \begin{equation}
          \begin{split}
      \int_{0}^{\infty} t^{\alpha-1} \mathrm{e}^{-c t} \pFq{1}{1}{a}{b}{-t} \mathrm{d}t &=
        c^{-\alpha} \Gamma(\alpha) \pFq{2}{1}{a, \alpha}{b}{-\frac{1}{c}} \quad
        \text{for } \alpha > 0 ~\mathrm{and}~ c > 0,
          \end{split}
          \label{eq_app_I-002-b}
      \end{equation}
        \begin{equation}
          \begin{split}
    &\left( (a-b) z + c - 2 a \right) \pFq{2}{1}{a, b}{c}{z} = 
    \left( c-a \right) \pFq{2}{1}{a-1, b}{c}{z} + 
    a \left( z-1 \right) \pFq{2}{1}{a+1, b}{c}{z},
          \end{split}
          \label{eq_app_I-002-a-ext2}
      \end{equation}
        \begin{equation}
          \begin{split}
    \int_{0}^{\infty} \mathrm{e}^{-(a x^2 + b x)}\mathrm{d}x &= 
    \frac{1}{2}\sqrt{\frac{\pi}{\alpha}} \exp\left( \frac{b^2}{4 a}\right) 
    \mathrm{erfc}\left( \frac{b}{2 \sqrt{a}} \right) \quad
        \text{for } a > 0 ~\mathrm{and}~ b > 0.
          \end{split}
          \label{eq_app_I-002-b-ext1}
      \end{equation}

The binomial coefficient can be defined for real constants $x, y$ using the gamma function as 
        \begin{align}
    \binom{x}{y} = \frac{\Gamma(x+1)}{\Gamma(y+1)\Gamma(x-y+1)},\quad
    \Gamma(t) = \int_{0}^{\infty}x^{t-1} \mathrm{e}^{-x} \mathrm{d}x.
        \label{eq_app_I-003}
        \end{align}
Appell's function $F_2\left( \Bigcdot \right)$ is defined via the Pochhammer symbol $(x)_{n} = \frac{\Gamma(x+n)}{\Gamma(x)}$ as follows
        \begin{align}
    &F_2\left( \alpha; \beta, \beta'; \gamma, \gamma'; x, y \right) = \sum_{m=0}^{\infty} \sum_{n=0}^{\infty}
    \frac{\left( \alpha \right)_{m+n} \left( \beta \right)_{m} \left( \beta' \right)_{n}}{m!~ n! \left( \gamma \right)_{m} \left( \gamma' \right)_{n}} x^m y^n.
    \label{eq_app_I-002-a-ext5}
        \end{align}
Appell's function can be reduced to the hypergeometric function using the following properties
  \begin{equation}
  \begin{split}
  F_2\left( d; a, a'; c, c'; 0, y \right) &= \pFq{2}{1}{d, a'}{c'}{y},\quad
  F_2\left( d; a, a'; c, c'; x, 0 \right) = \pFq{2}{1}{d, a}{c}{x}.
  \end{split}
  \label{eq_app_I-002-a-ext6}
  \end{equation}
The following integration holds under the following constraints $d > 0$ and $|k| + |k|' < |h|$  
\begin{equation}
  \begin{split}
  \int_{0}^{\infty} t^{d-1} \mathrm{e}^{-h t} \pFq{1}{1}{a}{b}{k t} \pFq{1}{1}{a'}{b'}{k' t} \mathrm{d}t = h^{-d} \Gamma(d) F_2\left( d; a, a'; b, b'; \frac{k}{h}, \frac{k'}{h} \right).
  \end{split}
  \label{eq_app_I-002-a-ext7}
  \end{equation}

The confluent Hypergeometric series of $n$ variables $\Psi_{2}\left( a; b_1, \ldots b_n ; x_1, \ldots, x_n \right)$ is defined in \cite[Volume 1, P385]{1579} as a $n$-fold infinite summation as written follows 
\begin{equation}
  \begin{split}
  \Psi_{2}\left( a; b_1, \ldots b_n ; x_1, \ldots, x_n \right) = \sum_{m_1, \ldots, m_n \geq 0} 
  \frac{(a)_{m_1+\ldots+m_n}}{(b_1)_{m_1} \cdots (b_n)_{m_n} m_1! \ldots m_n!} x_1^{m_1} \cdots x_n^{m_n},
  \end{split}
  \label{q004-eq_app_I-001}
  \end{equation}
The $n$-th order Bessel function of the first kind $\BesselJ{n}{x}$ has the following functional properties
      \begin{equation}
        \begin{split}
          \int_{0}^{t} x \BesselJ{0}{a x} \mathrm{d}x = \frac{1}{a} t \BesselJ{1}{a t} 
          \quad
        \text{for } t \geq 0,
        \end{split}
        \label{eq_app_I-exo1-001}
      \end{equation}
      \begin{equation}
        \begin{split}
        &\int_{0}^{\infty} t^{\nu-1} \mathrm{e}^{-p t} \cdot \prod_{i=1}^{n} \BesselJ{2 \mu_i}{2 \sqrt{a_i t}} \mathrm{d}t\\ 
        = &\frac{\Gamma(\nu+M)}{p^{\nu+M}} \left( \prod_{i=1}^{n} \frac{a_i^{\mu_i}}{\Gamma(2 \mu_i+1)} \right)
        \Psi_2\left( \nu+M; 2\mu_1+1, \ldots, 2\mu_n+1; \frac{a_1}{p}, \ldots, \frac{a_n}{p} \right),
        \end{split}
        \label{eq_app_I-exo1-002}
      \end{equation}
where $M = \sum_{i=1}^{n} \mu_i$ and (\ref{eq_app_I-exo1-002}) holds for $\operatorname{Re}(\nu+M) >0,~ \operatorname{Re}(p) >0$ \cite[Volume 1, P187, (43)]{1579}. 

Gauss-Laguerre quadratures can be used to evaluate the following integral for a given analytic function $g(x)$ as 
        \begin{align}
    \int_{0}^{\infty} \mathrm{e}^{-x} g(x) \mathrm{d}x = \sum_{n=1}^{N} w_n g\left( x_n \right) + R_N,
    \label{eq_app_I-002-a-ext9}         
        \end{align}
  where $x_n$ and $w_n$ are the $n$-th abscissa and weight of the $N$-th order Laguerre polynomial.

\section*{Appendix II}

In this appendix, we provide a proof of Theorem 1. First, we note that the in-phase and quadrature component of the received signal $\tilde{V}$ are circularly symmetric \cite{Durgin2002, Durgin2002a, Salo2006, Yu2007} due to the underlying assumption of uniformly distributed phase $\phi_i$ and $\phi_d$. Then, the PDF of the envelope $R$ and the joint characteristic function (CF) $\Phi(\nu)$ of the complex signal $\tilde{V}$ can be written as a Hankel transform pair as follows \cite{Durgin2002a, Salo2006}
\begin{equation}
  \begin{split}
  f_{R}(r) &= r \int_{0}^{\infty} \nu \Phi(\nu) \BesselJ{0}{r \nu} \mathrm{d}\nu,\quad
  \Phi(\nu) = \int_{0}^{\infty} f_{R}(r) \BesselJ{0}{\nu r} \mathrm{d}r.
  \end{split}
  \label{q004-eq_app_II-001}
\end{equation}
For the GMR-U model in (\ref{q004-001}), the CF $\Phi(\nu)$ can be derived as follows \cite{Abdi2000}
\begin{equation}
  \begin{split}
  \Phi(\nu) &= \mathbb{E}_{V_1, \ldots, V_N, V_d}\left[ \prod_{i=1}^{N} \BesselJ{0}{V_i \nu} \cdot \BesselJ{0}{V_d \nu} \right] = \prod_{i=1}^{N} \BesselJ{0}{V_i \nu} \mathbb{E}_{V_d}\left[ \BesselJ{0}{V_d \nu} \right]\\
  &= \Hypergeometric{1}{1}{m}{1}{-\frac{\Omega}{4m} \nu^2} \cdot \prod_{i=1}^{N} \BesselJ{0}{V_i \nu}, 
  \end{split}
  \label{q004-eq_app_II-002}
\end{equation}
where the second equality is due to the deterministic nature of $V_i$'s and the last expression follows by substituting (\ref{q004-ext-000}) and applying \cite[6.631]{Gradshteyn1994}. 

Next, we utilize the following series expansion of the confluent Hypergeometric function
\begin{equation}
  \begin{split}
  \Hypergeometric{1}{1}{m}{1}{-\frac{\Omega}{4m} \nu^2} &=\exp\left( -\frac{\Omega}{4m} \nu^2 \right) 
  L_{m-1}\left(\frac{\Omega}{4m} \nu^2\right)\\ &= \exp\left( -\frac{\Omega}{4m} \nu^2 \right)  
  \sum_{k = 0}^{m-1} \binom{m-1}{k} \frac{(-1)^k}{k!} \left( \frac{\Omega}{4m} \nu^2 \right)^k 
  \quad \text{for } m \in \mathbb{Z}^{+},
  \end{split}
  \label{q004-eq_app_II-003}
\end{equation}
where we applied \cite[07.20.03.0108.01]{wolfram} to the first equality and used (\ref{eq_app_I-001-a}) in the last equality. By substituting (\ref{q004-eq_app_II-002}) and (\ref{q004-eq_app_II-003}) to (\ref{q004-eq_app_II-001}), the PDF of the envelope $R$ can be expressed as follows 
\begin{equation}
  \begin{split}
  &f_{R}(r) = r \sum_{k = 0}^{m-1} \binom{m-1}{k} \frac{(-1)^k}{k!} \left( \frac{\Omega}{4m} \nu^2 \right)^k  \underbrace{\int_{0}^{\infty} \nu^{2k+1} \exp\left( -\frac{\Omega}{4m} \nu^2 \right) \prod_{i=1}^{N} \BesselJ{0}{V_i \nu} \cdot \BesselJ{0}{r \nu} \mathrm{d}\nu}_{\triangleq I_1},
  \end{split}
  \label{q004-eq_app_II-004}
\end{equation}
where the integral $I_1$ can be evaluated in a closed form expression as follows
\begin{equation}
  \begin{split}
  I_1 &= \frac{1}{2} \int_{0}^{\infty} t^{k} ~\mathrm{e}^{-\frac{\Omega}{4m} t} \cdot \prod_{i=1}^{N} \BesselJ{0}{V_i t^{1/2}} \BesselJ{0}{r t^{1/2}} \mathrm{d}t\\
  &= \frac{\Gamma(k+1)}{2} \left( \frac{4m}{\Omega} \right)^{k+1} 
  \Psi_{2}\left( k+1; [1]_{N+1};  
  \frac{V_1^2 m}{\Omega}, \ldots, \frac{V_N^2 m}{\Omega}, \frac{r^2 m}{\Omega}
  \right),
  \end{split}
  \label{q004-eq_app_II-005}
\end{equation}
by employing a change of variable, \textit{i.e.}, $t \leftarrow \nu^2$, and applying (\ref{eq_app_I-exo1-002}) to (\ref{q004-eq_app_II-005}). Hence, by substituting (\ref{q004-eq_app_II-005}) to (\ref{q004-eq_app_II-004}), (\ref{q004-004}) follows readily. 

The corresponding CDF of the signal envelope $R$ can be evaluated as below
\begin{equation}
  \begin{split}
  F_{R}(t) &= \int_{0}^{t} f_{R}(r) \mathrm{d}r = \int_{0}^{\infty} \nu 
  \Hypergeometric{1}{1}{m}{1}{-\frac{\Omega}{4m} \nu^2} \prod_{i=1}^{N} \BesselJ{0}{V_i \nu} 
  \left( \int_{0}^{t} r \BesselJ{0}{r \nu} \mathrm{d}r \right) \mathrm{d}\nu \\
  &= t \int_{0}^{\infty} 
  \Hypergeometric{1}{1}{m}{1}{-\frac{\Omega}{4m} \nu^2} \prod_{i=1}^{N} \BesselJ{0}{V_i \nu} 
  \cdot \BesselJ{1}{t \nu} \mathrm{d}\nu,
    \end{split}
  \label{q004-eq_app_II-006}
\end{equation}
by substituting (\ref{q004-004}) to the first equality, changing the order of integration in the second equality
and applying (\ref{eq_app_I-exo1-001}) in the third equality. The closed form expression of (\ref{q004-eq_app_II-006}) can be easily derived by following the similar procedure as the PDF of $R$, \textit{i.e.}, use the series expansion from (\ref{q004-eq_app_II-003}) and apply (\ref{eq_app_I-exo1-002}) with a change of variable $t \leftarrow \nu^2$. This completes the proof.

\section*{Appendix III}

In this appendix, we provide a proof of Theorem 2. The PDF of the FMR model is given by 
\begin{equation}
  \begin{split}
  f_{R}(r) &= \int_{0}^{\infty} f_{R|\xi}(r|u) f_{\xi}(u) \mathrm{d}u.
  \end{split}
  \label{q004-eq_app_III-001}
\end{equation}
Since the FMR model at a given $\xi$ follows the GMR-U model, (\ref{q004-eq_app_III-001}) can be written as below
\begin{equation}
  \begin{split}
  f_{R}(r) &= r \int_{0}^{\infty} \Hypergeometric{1}{1}{m}{1}{-\frac{\Omega}{4m} \nu^2} \cdot 
  \nu \BesselJ{0}{r \nu} \cdot \mathbb{E}_{\xi}\left[ \prod_{i=1}^{N} \BesselJ{0}{\sqrt{\xi} V_i \nu} \right]
	\mathrm{d}\nu\\
	&= \frac{m_s^{m_s}}{\Gamma(m_s)} 
	r \int_{0}^{\infty} \Hypergeometric{1}{1}{m}{1}{-\frac{\Omega}{4m} \nu^2} 
  \nu \BesselJ{0}{r \nu} \left( \int_{0}^{\infty} u^{m_s-1} \mathrm{e}^{-m_s u} 
	\prod_{i=1}^{N} \BesselJ{0}{\sqrt{u} V_i \nu} \mathrm{d}u \right)
	\mathrm{d}\nu\\
	&= r \int_{0}^{\infty} \Hypergeometric{1}{1}{m}{1}{-\frac{\Omega}{4m} \nu^2}
      \nu \BesselJ{0}{r \nu} \cdot 
      \Psi_{2}\left( m_s; [1]_{N}; \underline{\beta}(\nu) \right) 
      \mathrm{d}\nu,
  \end{split}
  \label{q004-eq_app_III-002}
\end{equation}
where we used (\ref{q004-004}) and changed the order of integration in the first equality, applied (\ref{q004-ext-001}) in the second equality and employed (\ref{eq_app_I-exo1-002}) in the last equality. The CDF in (\ref{q004-010}) can be readily achieved by using similar procedures as (\ref{q004-eq_app_II-006}), \textit{i.e.}, change the order of integration and apply (\ref{eq_app_I-exo1-001}). This completes the proof.

\section*{Appendix IV}

In this appendix, we provide a proof of Theorem 3. It is known that an arbitrary PDF of a positive random variable can be expressed in a series form using its higher order moments and the Laguerre polynomial as follows \cite{Abdi2000,Chai2009, Chun2016b}
\begin{equation}
  \begin{split}
  f_R(r) &= 2 \epsilon r \exp\left( -\epsilon r^2 \right) \cdot \sum_{n=0}^{\infty} C_n L_n\left( \epsilon r^2 \right),
  \end{split}
  \label{q004-eq_app_IV-001}
\end{equation}
where $L_n(x)$ is the Laguerre polynomial, $\epsilon$ is a positive, real valued constant\footnote{
	The constant $\epsilon$ was originally introduced in \cite{Abdi2000} as a tunable parameter that can minimize the approximation error of (\ref{q004-eq_app_IV-001}). Recently, the authors in \cite{Chai2009} proposed the optimal $\epsilon$ as $\epsilon = \left({\sum_{i=1}^{N} V_i^2 + \Omega}\right)^{-1}$ and numerically validated their proposition. We adopt this proposition and assume $\epsilon = \left({\sum_{i=1}^{N} V_i^2 + \Omega}\right)^{-1}$ in this paper. 
	} 
that can control the approximation accuracy, the coefficient $C_n$ denotes the following expression
\begin{equation}
  \begin{split}
  C_n &= \mathbb{E}_{R}\left[ L_n(\epsilon r^2) \right] 
	= \sum_{k=0}^{n} \frac{\left( -\epsilon \right)^k}{k!} \binom{n}{k} \mathbb{E}_{R}[r^{2k}],
	  \end{split}
  \label{q004-eq_app_IV-002}
\end{equation}
using (\ref{eq_app_I-001-a}) in the last equality. The corresponding series form of the CDF can be obtained by integrating (\ref{q004-eq_app_IV-001}) from zero to $t$ and applying (\ref{eq_app_I-001-b}), which achieves (\ref{q004-IIIB-001}). 

To evaluate (\ref{q004-eq_app_IV-002}), we use the following notation of the even moments and adopt the recursive formulation that was proposed in \cite{Goldman1972} as descried below
\begin{equation}
	\begin{split}
		\nu_{i}^{(2k)} = 
		\begin{dcases}
		\mathbb{E}\left[ \left( X_i^2 + Y_i^2 \right)^k\right] \hfill &\text{for } i = 1, \ldots, N,\\
		\mathbb{E}\left[ V_d^{2k} \right] \hfill &\text{for } i = N+1, \text{ which is the diffuse component},\\
		\end{dcases}
	\end{split}
	\label{q004-eq_app_IV-003}
\end{equation}
where $X_i$ and $Y_i$ represents the real and imaginary term of the $i$-th specular component. The magnitude of each specular component is constant for the MWGD model, whereas for the FMR model, $\nu_{i}^{(2k)}$ requires the $k$-th order moments of the Gamma distribution 
\begin{equation}
	\begin{split}
		\nu_{i}^{(2k)}~\text{for } i = 1, \ldots, N = 
		\begin{dcases}
		V_i^{2k} \hfill &\text{for MWGD model},\\
		\mathbb{E}_{\xi}\left[ \mathbb{E}\left[ \left( \sqrt{\xi} V_i\right)^{2k}\bigg\rvert\xi \right] \right] 
		= \mathbb{E}\left[\xi^k\right] \cdot V_i^{2k}  \hfill &\text{for FMR model},
		\end{dcases}
	\end{split}
	\label{q004-eq_app_IV-004}
\end{equation}
which can be evaluated by using (\ref{q004-ext-001}) and the notion of Gamma function as follows
\begin{equation}
	\begin{split}
	\mathbb{E}\left[\xi^k\right] &= \frac{{m_s}^{m_s}}{\Gamma({m_s})}
	\int_{0}^{\infty}  u^{{m_s+k}-1} \mathrm{e}^{-{m_s} u} \mathrm{d}u = 
	\frac{1}{\Gamma({m_s}) m_s^k} 
	\int_{0}^{\infty}  t^{{m_s+k}-1} \mathrm{e}^{-t} \mathrm{d}t = 
	\frac{\left(m_s\right)_k}{m_s^k}. 
	\end{split}
	\label{q004-eq_app_IV-005}
\end{equation}
Similarly, the moment $\nu_{N+1}^{(2k)}$ of the diffuse component can be evaluated as follows
\begin{equation}
	\begin{split}
	\mathbb{E}\left[V_d^{2k}\right] &= \frac{2}{\Gamma(m)}\left( \frac{m}{\Omega} \right)^m 
	\int_{0}^{\infty} 
	x^{2k+2m -1} \exp\left(-\frac{m}{\Omega} x^2\right) \mathrm{d}x\\
	&= \frac{1}{\Gamma({m})} \left( \frac{\Omega}{m} \right)^k 
	\int_{0}^{\infty} t^{k+m-1} \mathrm{e}^{-t} \mathrm{d}t = (m)_k \left( \frac{\Omega}{m} \right)^k, 
	\end{split}
	\label{q004-eq_app_IV-006}
\end{equation}
where we substituted (\ref{q004-ext-000}) in the first equality, used a change of variable, \textit{i.e.}, $t \leftarrow \frac{m}{\Omega} x^2$, in the second equality and applied (\ref{eq_app_I-003}) in the last equality. This completes the proof.

\bibliographystyle{IEEEtran}
\bibliography{bib1}

\end{document}